\documentclass[letter,10pt]{article}
\usepackage[utf8]{inputenc}

\usepackage{amsmath, amssymb, amsthm, graphicx}
\usepackage{multicol}

\usepackage{url}

\newcommand{\blankSpace}[0]{\rule{30mm}{.4pt}}

\newcommand{\ruleset}[1]{\textsc{#1}}

\newcommand{\cproblem}[1]{\ensuremath{\textsc{#1}}} 
\newcommand{\cclass}[1]{\ensuremath{\mathord{\rm #1}}} 

{\vspace*{-0.5em}\noindent {\bf Proof.}
}
{$\Box$ \bigskip}

\newcommand{\BigAnd}{\displaystyle\bigwedge}

\newtheorem{definition}{Definition}

\newtheorem{theorem}{Theorem}
\newtheorem{lemma}{Lemma}

\title{$2^3$ Quantified Boolean Formula Games and Their Complexities}
\author{Kyle Burke\\Colby College\\kgburke@colby.edu}

\begin{document}

\maketitle

\begin{abstract}
  Consider \cproblem{QBF}, the Quantified Boolean Formula problem, as a combinatorial game ruleset.  The problem is rephrased as determining the winner of the game where two opposing players take turns assigning values to Boolean variables.  In this paper, three variations of games are applied to create seven new games: whether each player is restricted to where they may play,  which values they may set variables to, or whether conditions they are shooting for at the end of the game differ.  The complexity for determining which player can win is analyzed for all games.  Of the seven, two are trivially in \cclass{P} and the other five are \cclass{PSPACE}-complete.  These varying properties are common for combinatorial games; reductions from these five hard games can simplify the process for showing the \cclass{PSPACE}-hardness of other games.
\end{abstract}

\section{Introduction}

\subsection{Combinatorial Game Theory}

Two-player games with alternating turns, perfect information, and no random elements are known as \emph{combinatorial games}.  Combinatorial Game Theory determines which of the players has a winning move from any \emph{position} (game state).  Many of the elegant aspects of CGT are ignored here, such as how to determine optimal moves and how to add two games together.  The interested reader is encouraged to browse \cite{WinningWays:2001} and \cite{LessonsInPlay:2007}.  

\subsection{Quantified 3SAT}

In computational complexity, the quantified Boolean satisfiability problem, \cproblem{QBF}, consists of determining whether formulas of the form are true.  $\exists x_0: \forall x_1: \exists x_2: \forall x_3: \cdots Q_{n-1} x_{n-1}: f(x_0, x_1, x_2, \ldots, x_{n-1})$, where $f$ is a Boolean formula and $Q_i = \begin{cases} \exists &, i \mbox{ is even} \\ \forall &, i \mbox{ is odd} \end{cases}$. \cproblem{QBF} is commonly considered the fundamental problem for the complexity class \cclass{PSPACE}, the set of all problems that can be solved using workspace polynomial in the size of the input.  In other words, \cclass{PSPACE} can be rephrased as the set of all problems that can be efficiently reduced to \cproblem{QBF}.

Each instance can be reconsidered as a combinatorial game between two players: Even/True and Odd/False.  The initial position consists of a list of the indexed literals and the unquantified formula, $f$.  Even/True moves first, choosing a value for $x_0$.  Odd/False goes next, choosing a value for $x_1$.  The players continue taking turns in this manner until all variables have been assigned.  At this point, the value of $f$ is determined.  If $f$ is true, Even/True wins, otherwise Odd/False wins.

\subsubsection{Sample QBF Game}
\label{section:sampleQBF}

For example, consider the position with formula: $(\overline{x_0} \vee x_3 \vee \overline{x_1}) \wedge (x_2 \vee x_1 \vee \overline{x_6}) \wedge (x_4 \vee \overline{x_6} \vee x_0) \wedge (\overline{x_2} \vee \overline{x_4} \vee x_3)$ and no assignments.  Then the following would be a legal sequence of turns:

\begin{itemize}
  \item  Even/True chooses True (\textbf{T}) for $x_0$.  The players are likely keeping track of the assignments and updating the formula:\\ 
    \begin{tabular}{l}
	  $(\textbf{F} \vee x_3 \vee \overline{x_1}) \wedge (x_2 \vee x_1 \vee \overline{x_6}) \wedge (x_4 \vee \overline{x_6} \vee \textbf{T}) \wedge (\overline{x_2} \vee \overline{x_4} \vee x_3)$ \\
      $=(x_3 \vee \overline{x_1}) \wedge (x_2 \vee x_1 \vee \overline{x_6}) \wedge (\overline{x_2} \vee \overline{x_4} \vee x_3)$
	\end{tabular}
  \item Next, Odd/False chooses \textbf{T} for $x_1$.  Formula:\\
	\begin{tabular}{l}
	  $(x_3 \vee \textbf{F}) \wedge (x_2 \vee \textbf{T} \vee \overline{x_6}) \wedge (\overline{x_2} \vee \overline{x_4} \vee x_3)$ \\
	  $ = x_3 \wedge (\overline{x_2} \vee \overline{x_4} \vee x_3)$
	\end{tabular}
  \item Even/True (not feeling very confident at this point), chooses \textbf{F} for $x_2$:\\
	\begin{tabular}{l}
	  $x_3 \wedge (\textbf{T} \vee \overline{x_4} \vee x_3)$ \\
	  $ = x_3$
	\end{tabular}
  \item Odd/False triumphantly chooses \textbf{F} for $x_3$.  With the assignments, the evaluated formula will be False in the end.  Despite this, the players may continue taking their turns:
  \item Even/True chooses \textbf{T} for $x_4$.
  \item Odd/False chooses \textbf{F} for $x_5$.  (Notice that no instances of $x_5$ exist in the formula.)
  \item Even/True chooses \textbf{F} for $x_6$.
  \item The assignments cause the formula to be false; Odd/False wins.
\end{itemize}

\subsection{Algorithmic Combinatorial Game Theory}

Notice that determining, from the initial position, whether the Even/True player can win the game is exactly the same problem as determining whether the \cproblem{QBF} instance is true.  Due to this equivalence, we will abuse notation slightly and refer both this ruleset and the computational problem as \cproblem{QBF}.  A position in this ruleset is the unquantified formula, $f$, and the list of $n$ indexed Boolean variables, with assignments for the first $k-1$ ($0 \leq k \leq n$).  

The lack of distinction between computational problems and combinatorial rulesets is not specific to \ruleset{QBF}.  It is common to use a ruleset's name to refer to the computational complexity of a ruleset by the induced problem of determining which player can win.  For example, we say that \ruleset{Nim} is in \cclass{P}.  (To be even more specific, \ruleset{Nim} is in $O(n)$ where $n$ is the number of heaps of sticks.)  The study of algorithms and computational complexity to determine the winner is known as algorithmic combinatorial game theory \cite{AlgGameTheory_GONC3}.  Much of this paper is concerned with the computational hardness for new games based on formula satisfiability.  

\subsection{Another SAT Game: \ruleset{Positive CNF}}
\label{section:positive-cnf}

Many other Boolean formula satisfiability games have been defined \cite{DBLP:journals/jcss/Schaefer78}.  We go into more detail for one of them, $G_{POS}$(POS CNF)---here referred to as \ruleset{Positive CNF}---as it will be used in later proofs.  Other satisfiability games from \cite{DBLP:journals/jcss/Schaefer78} are highly recommended to the interested reader.

\begin{definition}[Positive CNF]
  \emph{\ruleset{Positive CNF}} is the ruleset for games played on a Boolean 3-CNF formula, $f$, using $n$ variables without including any negations.  The players are indicated True and False.  Each turn, the current player chooses any one unassigned variable and assigns it the value corresponding to their name.  When all variables are assigned, True wins if the value of the formula is true, otherwise False wins.
\end{definition}

\ruleset{Positive CNF} is known to be \cclass{PSPACE}-complete \cite{DBLP:journals/jcss/Schaefer78}, which will make two of the five reductions (sections \ref{section:by-player-anywhere-different} and \ref{section:either-anywhere-different}) very simple.

\section{Three Ruleset Toggles}

Some \cclass{PSPACE}-hard combinatorial games are difficult to reduce to from \cproblem{QBF}.  These target games often have properties---such as being able to play anywhere on the ``board''---very dissimilar from \cproblem{QBF}.  There are three common properties of games that we vary to modify \cproblem{QBF}.  We refer to each property as a ruleset \emph{toggle}, since each has exactly two possible values.

\subsection{Play Location}

In \cproblem{QBF}, on the $i^{th}$ turn, the current player assigns a value to variable $x_i$.  Many combinatorial games are more flexible, allowing the next player to play wherever on the board they would like.  We can model that by allowing a variant of \ruleset{QBF} where the current player may choose to play at any unassigned variable.  We refer to this toggle as \emph{locality} with possible values \emph{local} and \emph{anywhere}. 

\begin{itemize}
  \item \textbf{local}: The current player assigns a value to the unassigned variable, $x_i$, with lowest index.
  \item \textbf{anywhere}: The current player assigns a value to any one of the unassigned variables.
\end{itemize}

\subsection{Boolean Choice}

In \cproblem{QBF}, each player chooses to set a variable either true or false.  However, in rulesets such as Snort, a player is identified by their color (either Red or Blue) and may only play pieces of that color.  We model this with the toggle \emph{Boolean identity}, with possible values \emph{either} and \emph{by player}.

\begin{itemize}
  \item \textbf{either}: The current player assigns either true or false to a variable.
  \item \textbf{by player}: One player only assigns variables to true, the other player only assigns false.  
\end{itemize}

\subsection{Goal}

In \ruleset{QBF}, one player is trying to force the formula to have the value true, while the other one is vying for false.  In all impartial games (and others) both players have the same goal; the one who makes the last move to reach that goal wins the game.  We model that with the \emph{goal} toggle, with values \emph{different} and \emph{same}.  When the goal is the same, players avoid creating a formulas with some property.  Unfortunately, we cannot simply use contradictions (or tautologies) as the illegal formulas.

\subsubsection{Contradictions are not Enough}

There are two reasons it is impractical to avoid contradictions:

\begin{itemize}
 \item Determining whether a boolean formula is satisfiable is NP-hard.  Thus, it would be computationally difficult to determine when the formula is a contradiction and the game is over.  Trying to enforce this would be extra onerous.  
 \item If contradictions are not legal positions, then each position's formula must be satisfiable.  That means all unset variables can be used, and the number of remaining turns in the game is known; determining who will win is trivial.
\end{itemize}

\subsubsection{Blatantly False}

Something a bit more complicated is necessary.  We define an alternative to contradictions: \emph{blatantly false} formulas.

\begin{definition}[Blatantly False]
  A \emph{blatantly false} formula is one which is either:
   \begin{itemize}
     \item a false-assigned literal, 
     \item the negation of a blatantly true subformula,
     \item the or of multiple subformulas, all of which are, recursively, blatantly false, or
     \item the and of multiple subformulas, one of which is, recursively, blatantly false.
   \end{itemize}
\end{definition}

  \emph{Blatantly True} is defined analagously:
	\begin{itemize}
	  \item a true-assigned literal,
	  \item the negation of a blatantly false subformula,
	  \item the or of multiple subformulas, one of which is, recursively, blatantly true, or
	  \item the and of multiple subformulas, all of which are, recursively, blatantly true.
	\end{itemize}

The following are two examples to illuminate this idea:

\begin{itemize}
  \item $\overline{x_i} \wedge \left(\textbf{F} \wedge (x_j \vee x_k)\right) \wedge x_l$ is blatantly false.
  \item $x_i \wedge \overline{x_i}$ is not blatantly false, even though it is a contradiction.  (It is easy to determine whether a formula is blatantly false, which is not necessarily the case with contradictions.)
\end{itemize}

With this notion, we can define the two possible values for the goal toggle:

\begin{itemize}
  \item \textbf{different}: One player is trying to set the formula to true, the other to false.  After all variables are assigned, the formula is evaluated.  Whichever player has reached their target value wins.
  \item \textbf{same}: Players are not allowed to assign a variable such that the resulting formula is blatantly false.  If a player cannot move, they lose the game---the usual end-of-game condition for combinatorial games.  This means that there are no legal moves on the formula: $x_0 \wedge \overline{x_0}$; the first player would automatically lose.  Also, if the formula is true at the end of the game, the last player to choose the value of a variable wins.
\end{itemize}

\subsection{QBF Categorized}

\ruleset{QBF}, then, is the \ruleset{either-local-different} ruleset: players have to play on the next variable, are allowed to set the value to either true or false, and one is shooting to make the formula true, while the other tries to make it false.  There are seven total other rulesets generated by changing these variables.  The following sections cover each of these and analyze their computational complexity.  These results are summarized in table \ref{table:main}.

\begin{table}
  \begin{center}
	\begin{tabular}{|c|c|c||c|c|}
	  \hline
	  Boolean choice & play location & goal & section & complexity \\ \hline \hline
	  by player & local & same & \ref{section:by-player-local} & In \cclass{P} \\ \hline
	  by player & local & different & \ref{section:by-player-local} & In \cclass{P} \\ \hline
	  by player & anywhere & same & \ref{section:by-player-anywhere-same} & \cclass{PSPACE}-complete\\ \hline
	  by player & anywhere & different & \ref{section:by-player-anywhere-different} & \cclass{PSPACE}-complete \\ \hline
	  either & local & same & \ref{section:either-local-same} & \cclass{PSPACE}-complete\\ \hline
	  either & local & different & none & \cclass{PSPACE}-complete\\ \hline
	  either & anywhere & same & \ref{section:either-anywhere-same} & \cclass{PSPACE}-complete\\ \hline
	  either & anywhere & different & \ref{section:either-anywhere-different} & \cclass{PSPACE}-complete\\ \hline
	\end{tabular}
  \end{center}
  \caption{QBF Ruleset Complexities}
  \label{table:main}
\end{table}

\section{By-Player-Local-X Rulesets}
\label{section:by-player-local}

Two of the rulesets generated are trivial for determining the winner: \ruleset{by-player-local-same} and \ruleset{by-player-local-different}.  In the games where players have neither a choice of the Boolean value nor the location to play, there is only one possible sequence of moves.  To figure out which player will win, a program needs only simulate the moves, adhering to the goal condition.   

For example, consider the \ruleset{by-player-local-same} initial position with the formula described in section \ref{section:sampleQBF}: $(\overline{x_0} \vee x_3 \vee \overline{x_1}) \wedge (x_2 \vee x_1 \vee \overline{x_6}) \wedge (x_4 \vee \overline{x_6} \vee x_0) \wedge (\overline{x_2} \vee \overline{x_4} \vee x_3)$.  Then the following would be the sequence of turns:

\subsection{Sample Game}

\begin{itemize}
  \item  Even/True must assign \textbf{T} to $x_0$.  We keep track by updating the formula:\\ 
    \begin{tabular}{l}
	  $(\textbf{F} \vee x_3 \vee \overline{x_1}) \wedge (x_2 \vee x_1 \vee \overline{x_6}) \wedge (x_4 \vee \overline{x_6} \vee \textbf{T}) \wedge (\overline{x_2} \vee \overline{x_4} \vee x_3)$ \\
      $=(x_3 \vee \overline{x_1}) \wedge (x_2 \vee x_1 \vee \overline{x_6}) \wedge (\overline{x_2} \vee \overline{x_4} \vee x_3)$
	\end{tabular}
  \item Next, Odd/False assigns \textbf{F} to $x_1$.  Formula:\\
	\begin{tabular}{l}
	  $(x_3 \vee \textbf{T}) \wedge (x_2 \vee \textbf{F} \vee \overline{x_6}) \wedge (\overline{x_2} \vee \overline{x_4} \vee x_3)$ \\
	  $ = (x_2 \vee \overline{x_6}) \wedge (\overline{x_2} \vee \overline{x_4} \vee x_3)$
	\end{tabular}
  \item Even/True assigns \textbf{T} to $x_2$:\\
	\begin{tabular}{l}
	  $(\textbf{T} \vee \overline{x_6}) \wedge (\textbf{F} \vee \overline{x_4} \vee x_3)$ \\
	  $ = \overline{x_4} \vee x_3$
	\end{tabular}
  \item Odd/False assigns \textbf{F} to $x_3$:\\
	\begin{tabular}{l}
	  $\overline{x_4} \vee \textbf{F}$ \\
	  $ = \overline{x_4}$
	\end{tabular}
  \item Even/True now cannot make a move.  (Their only normal option: assigning \textbf{T} to $x_4$, would cause the formula to be blatantly false, so they cannot move to that position.)  Odd/False wins!
\end{itemize}

In \ruleset{by-player-local-different} with this same formula, Odd/False also wins because the formula evaluates to false by the end.

\section{By-Player-Anywhere-Same}
\label{section:by-player-anywhere-same}

The \ruleset{by-player-anywhere-same} QBF ruleset consists of the games with two players, True and False, each avoiding creating an blatantly false formula, while allowed to play anywhere on the board.  

\subsection{Sample Game}

For example, consider the initial position with the same formula given in section \ref{section:sampleQBF}: $(\overline{x_0} \vee x_3 \vee \overline{x_1}) \wedge (x_2 \vee x_1 \vee \overline{x_6}) \wedge (x_4 \vee \overline{x_6} \vee x_0) \wedge (\overline{x_2} \vee \overline{x_4} \vee x_3)$.  The following is a legal sequence of plays from the initial position:

\begin{itemize}
  \item True chooses $x_3$.  The partially-evaluated formula now looks like:\\
    \begin{tabular}{l}
      $(\overline{x_0} \vee \textbf{T} \vee \overline{x_1}) \wedge (x_2 \vee x_1 \vee \overline{x_6}) \wedge (x_4 \vee \overline{x_6} \vee x_0) \wedge (\overline{x_2} \vee \overline{x_4} \vee \textbf{T})$ \\
      $ = (x_2 \vee x_1 \vee \overline{x_6}) \wedge (x_4 \vee \overline{x_6} \vee x_0)$ \\
    \end{tabular}
  \item False chooses $x_1$.  (Notice that choosing $x_6$ would cost False the game since there would be an odd number of moves remaining.) \\
    \begin{tabular}{l}
      $(x_2 \vee \textbf{F} \vee \overline{x_6}) \wedge (x_4 \vee \overline{x_6} \vee x_0)$ \\
      $ = (x_2 \vee \overline{x_6}) \wedge (x_4 \vee \overline{x_6} \vee x_0)$
	\end{tabular}
  \item True chooses $x_2$:\\
    \begin{tabular}{l}
      $ (\textbf{T} \vee \overline{x_6}) \wedge (x_4 \vee \overline{x_6} \vee x_0)$ \\
      $ = x_4 \vee \overline{x_6} \vee x_0$
	\end{tabular}
  \item False chooses $x_4$:\\
    \begin{tabular}{l}
      $\textbf{F} \vee \overline{x_6} \vee x_0$ \\
      $\overline{x_6} \vee x_0$
	\end{tabular}
  \item True chooses $x_0$:\\
    \begin{tabular}{l}
      $\overline{x_6} \vee \textbf{T}$ \\
      $ = \textbf{T}$ \\
	\end{tabular}
	
	The formula now always evaluates to True.  (True is going to win because there are an even number of moves left.)
  \item False chooses $x_6$.
  \item True chooses $x_5$.  There are no more variables to assign, so True has made the last move and wins!
\end{itemize}

\subsection{\cclass{PSPACE}-completeness}

\begin{theorem}[\ruleset{by-player-anywhere-same} \cclass{PSPACE}-completeness.]
  Ruleset \ruleset{by-player-anywhere-same} is \cclass{PSPACE}-complete.  
  
  \begin{proof}
  To show hardness, we reduce from the well-known \cclass{PSPACE}-complete ruleset \ruleset{Snort} \cite{DBLP:journals/jcss/Schaefer78}.  A \ruleset{Snort} position consists of a graph, with some vertices painted blue, some red, and the rest uncolored.  The two players, Blue and Red, take turns choosing an uncolored vertex and painting it their own color.  Players are not allowed to paint vertices adjacent to the opposite color.  The first player who cannot play loses.
  
	The reduction is as follows.  Let $G = (V, E)$ be the \ruleset{Snort} graph and let $V = \{0, \ldots,n-1\}$.  The set of literals for our formula is $x_0, x_1, \ldots, x_{n-1}$.  Blue corresponds to True and Red corresponds to False.  For any edge, $(i, j) \in E$, we include two clauses: $(x_i \vee \overline{x_j}) \wedge (\overline{x_i} \vee x_j)$.  In the case where any of $i$ and $j$ are painted, the resulting formula is blatantly false only when they have opposing colors.  Thus, the only moves a player is not allowed to make correspond exactly to illegal moves in \ruleset{Snort}.  
	
	To get the overall formula, we conjoin all pieces together: $\BigAnd_{(i, j) \in E} (x_i \vee \overline{x_j}) \wedge (\overline{x_i} \vee x_j)$
	
	Since the winnability of the two games are equal, the \ruleset{by-player-anywhere-same} ruleset is \cclass{PSPACE}-hard.
  \end{proof}
\end{theorem}

\section{By-Player-Anywhere-Different}
\label{section:by-player-anywhere-different}

The \ruleset{by-player-anywhere-different} ruleset consists of games between two players, True and False, who may choose to play on any unassigned variable on their turn.  Each player wins if the value of the formula matches their identity.

\subsection{Sample Game}

Consider the initial position with the same formula given in section \ref{section:sampleQBF}: $(\overline{x_0} \vee x_3 \vee \overline{x_1}) \wedge (x_2 \vee x_1 \vee \overline{x_6}) \wedge (x_4 \vee \overline{x_6} \vee x_0) \wedge (\overline{x_2} \vee \overline{x_4} \vee x_3)$.  The following is a legal sequence of plays from the initial position:

\begin{itemize}
  \item True chooses $x_3$:\\
    \begin{tabular}{l}
      $(\overline{x_0} \vee \textbf{T} \vee \overline{x_1}) \wedge (x_2 \vee x_1 \vee \overline{x_6}) \wedge (x_4 \vee \overline{x_6} \vee x_0) \wedge (\overline{x_2} \vee \overline{x_4} \vee \textbf{T})$ \\
      $ = (x_2 \vee x_1 \vee \overline{x_6}) \wedge (x_4 \vee \overline{x_6} \vee x_0)$
	\end{tabular}
  \item False chooses $x_2$:\\
	\begin{tabular}{l}
	  $ = (\textbf{F} \vee x_1 \vee \overline{x_6}) \wedge (x_4 \vee \overline{x_6} \vee x_0)$ \\
	  $ = (x_1 \vee \overline{x_6}) \wedge (x_4 \vee \overline{x_6} \vee x_0)$
	\end{tabular}
  \item True chooses $x_1$:\\
	\begin{tabular}{l}
	  $ = (\textbf{T} \vee \overline{x_6}) \wedge (x_4 \vee \overline{x_6} \vee x_0)$\\
	  $ = (x_4 \vee \overline{x_6} \vee x_0)$
	\end{tabular}
  \item False chooses $x_4$:\\
	\begin{tabular}{l}
	  $ = \textbf{F} \vee \overline{x_6} \vee x_0$\\
	  $ = \overline{x_6} \vee x_0$
	\end{tabular}
  \item True chooses $x_0$:\\
	\begin{tabular}{l}
	  $ = \overline{x_6} \vee \textbf{T}$ \\
	  $ = \textbf{T}$
	\end{tabular}
  \item No matter which order $x_5$ and $x_6$ are chosen, True has won the game.
\end{itemize}

\subsection{\cclass{PSPACE}-completeness}

\begin{theorem}[\ruleset{by-player-anywhere-different} \cclass{PSPACE}-completeness.]
  Ruleset \ruleset{by-player-anywhere-different} is also \cclass{PSPACE}-complete.  
\end{theorem}
  
  \begin{proof}
	To show hardness, we reduce from \ruleset{positive CNF} (see Section \ref{section:positive-cnf}).  

	The rules for \ruleset{positive CNF} are exactly the same as \ruleset{by-player-anywhere-different}, except that not all formulas are allowed.  Thus, each instance of \ruleset{positive CNF} is also an instance of \ruleset{by-player-anywhere-different}, meaning the new game is ``automatically'' \cclass{PSPACE}-hard.
  \end{proof}

\section{Either-Local-Same}
\label{section:either-local-same}

The \ruleset{either-local-same} ruleset consists of the games between two players forced to play on a specific literal each turn while avoiding creating an blatantly false formula.  

\subsection{Sample Game}

As an example, consider the initial position with the same formula given in section \ref{section:sampleQBF}: $(\overline{x_0} \vee x_3 \vee \overline{x_1}) \wedge (x_2 \vee x_1 \vee \overline{x_6}) \wedge (x_4 \vee \overline{x_6} \vee x_0) \wedge (\overline{x_2} \vee \overline{x_4} \vee x_3)$.  The following is a legal sequence of plays from the initial position:

\begin{itemize}
  \item Even assigns \textbf{F} to $x_0$:\\
    \begin{tabular}{l}
      $(\textbf{T} \vee x_3 \vee \overline{x_1}) \wedge (x_2 \vee x_1 \vee \overline{x_6}) \wedge (x_4 \vee \overline{x_6} \vee \textbf{F}) \wedge (\overline{x_2} \vee \overline{x_4} \vee x_3)$ \\
      $ = (x_2 \vee x_1 \vee \overline{x_6}) \wedge (x_4 \vee \overline{x_6}) \wedge (\overline{x_2} \vee \overline{x_4} \vee x_3)$
	\end{tabular}
  \item Odd assigns \textbf{F} to $x_1$:\\
	\begin{tabular}{l}
      $(x_2 \vee \textbf{F} \vee \overline{x_6}) \wedge (x_4 \vee \overline{x_6}) \wedge (\overline{x_2} \vee \overline{x_4} \vee x_3)$ \\
      $ = (x_2 \vee \overline{x_6}) \wedge (x_4 \vee \overline{x_6}) \wedge (\overline{x_2} \vee \overline{x_4} \vee x_3)$
	\end{tabular}
  \item Even assigns \textbf{T} to $x_2$:\\
	\begin{tabular}{l}
      $(\textbf{T} \vee \overline{x_6}) \wedge (x_4 \vee \overline{x_6}) \wedge (\textbf{F} \vee \overline{x_4} \vee x_3)$ \\
      $ = (x_4 \vee \overline{x_6}) \wedge (\overline{x_4} \vee x_3)$
	\end{tabular}
  \item Odd assigns \textbf{F} to $x_3$:\\
	\begin{tabular}{l}
	  $(x_4 \vee \overline{x_6}) \wedge (\overline{x_4} \vee \textbf{F})$ \\
	  $ = (x_4 \vee \overline{x_6}) \wedge \overline{x_4}$ 
	\end{tabular}
  \item Even assigns \textbf{F} to $x_4$:\\
	\begin{tabular}{l}
	  $(\textbf{F} \vee \overline{x_6}) \wedge \textbf{T}$ \\
	  $\overline{x_6}$
	\end{tabular}
  \item Odd assigns \textbf{T} to $x_5$.  The formula does not change.
  \item Even assigns \textbf{F} to $x_6$.  The formula will now always evaluate to True.  There are no more variables, so Even wins!
\end{itemize}
  
\subsection{\cclass{PSPACE}-completeness}

\begin{theorem}[\ruleset{either-local-same} \cclass{PSPACE}-completeness.]
  Ruleset \ruleset{either-local-same} is also \cclass{PSPACE}-complete. 
\end{theorem}
  
  \begin{proof}
	To show hardness, we reduce from \ruleset{QBF}.  
  
    Assume our \ruleset{QBF} formula is written in conjunctive normal form (this subset of positions is still \cclass{PSPACE}-complete \cite{DBLP:journals/ipl/AspvallPT79}) using $n$ variables $x_0, \ldots, x_{n-1}$.  Let $\BigAnd_{i \in [c]} \varphi_i$ be the formula, with clauses $\varphi_0, \ldots, \varphi_{c-1}$.  
    
    For each $\varphi_i$, we create a new clause, $\gamma_i$, in the following way.  Let $l$ be the largest index of the literals in $\varphi_i$.  Now, $\gamma_i = \begin{cases}
                                        \varphi_i &, l \mbox{ is even} \\
                                        \left(\varphi_i \vee (x_{l+1} \wedge \overline{x_{l+1}})\right) &, l \mbox{ is odd}
	\end{cases}$.  Also, let $m = \begin{cases}
	  n+1 &, n \mbox{ is even} \\
	  n+2 &, n \mbox{ is odd}
	\end{cases}$.  In this way, $m$ will always be odd.
	
	The resulting position for \ruleset{either-local-same} consists of a formula with $c$ clauses, $\BigAnd_{i \in [c]} \gamma_i$, and has $m$ variables $x_0, \ldots, x_{m-1}$.  In some cases, no literal with index $m-1$ will appear in the formula.
	
	It remains to be shown that the winnability of the \ruleset{either-local-same} position is equivalent to the \ruleset{QBF} position.  Notice that in \ruleset{QBF}, if any one clause is blatantly false, the even/true player loses.  The $\gamma$-clauses simulate this in the reduced formula: the last assignment to a literal in a clause is always made by the even player.  If the clause would become blatantly false, even cannot move and loses.
	
	Alternatively, in order for odd to lose the game, the variables must have been set so that all clauses are true.  To maintain this condition, we just make sure the last variable index is even.  Thus, if all variables are successfully set, the even player will win the game after having made the last move.  To accomplish this, we enforce that $m$ must be odd, even if this means the literal $x_{m-1}$ does not appear in the reduced formula.
	
	Since the winnability of the two games are equal, the \ruleset{either-local-same} ruleset is \cclass{PSPACE}-hard.
  \end{proof}

\section{Either-Anywhere-Same}
\label{section:either-anywhere-same}

The \ruleset{either-anywhere-same} ruleset consists of the games between two players who can play either value on any unassigned variable each turn while avoiding creating an blatantly false formula.  

\subsection{Sample Game}

As an example, consider the initial position with the same formula given in section \ref{section:sampleQBF}: $(\overline{x_0} \vee x_3 \vee \overline{x_1}) \wedge (x_2 \vee x_1 \vee \overline{x_6}) \wedge (x_4 \vee \overline{x_6} \vee x_0) \wedge (\overline{x_2} \vee \overline{x_4} \vee x_3)$.  The following is a legal sequence of plays from the initial position:

\begin{itemize}
  \item Even assigns \textbf{F} to $x_6$:\\
    \begin{tabular}{l}
      $(\overline{x_0} \vee x_3 \vee \overline{x_1}) \wedge (x_2 \vee x_1 \vee \textbf{T}) \wedge (x_4 \vee \textbf{T} \vee x_0) \wedge (\overline{x_2} \vee \overline{x_4} \vee x_3)$ \\
      $ = (\overline{x_0} \vee x_3 \vee \overline{x_1})\wedge (\overline{x_2} \vee \overline{x_4} \vee x_3)$
    \end{tabular}
  \item Odd assigns \textbf{T} to $x_2$:\\
	\begin{tabular}{l}
	  $(\overline{x_0} \vee x_3 \vee \overline{x_1})\wedge (\textbf{F} \vee \overline{x_4} \vee x_3)$ \\
	  $ = (\overline{x_0} \vee x_3 \vee \overline{x_1})\wedge (\overline{x_4} \vee x_3)$
	\end{tabular}
  \item Even assigns \textbf{T} to $x_3$:\\
	\begin{tabular}{l}
	  $ = (\overline{x_0} \vee \textbf{T} \vee \overline{x_1})\wedge (\overline{x_4} \vee \textbf{T})$ \\
	  $\textbf{T}$ \\
	  From this point on, the formula will always evaluate to True.
	\end{tabular}
  \item With four variables remaining to assign and no chance the formula will be blatantly false, Even takes the last turn and wins.
\end{itemize}

\subsection{\cclass{PSPACE}-completeness}

\begin{theorem}[\ruleset{either-anywhere-same} \cclass{PSPACE}-completeness.]
  Ruleset \ruleset{either-anywhere-same} is also \cclass{PSPACE}-complete.  
\end{theorem}
  
  \begin{proof}
	The reduction here is similar to the reduction from \ruleset{Snort} to \ruleset{by-player-anywhere-same} in section \ref{section:by-player-anywhere-same}.  Instead of \ruleset{Snort}, we'll reduce from \ruleset{Proper 2-Coloring}, a different graph game where players take turns painting uncolored vertices so that no two neighboring vertices have the same color.  In this game, both players can choose either color on their turn.  \ruleset{Proper 2-coloring} is impartial and \cclass{PSPACE}-complete \cite{DBLP:journals/tcs/BeaulieuBD13}.  
	
	The reduction is as follows.  Let $G = (V, E)$ be the \ruleset{Proper 2-coloring} graph and let $V = \{0, \ldots, n-1\}$.  The literals for the formula are: $x_0, x_1, \ldots, x_{n-1}$.  Blue corresponds to True and Red corresponds to False.  Now, for each edge, $(i,j) \in E$, we include: $(x_i \wedge \overline{x_j}) \vee (\overline{x_i} \wedge x_j)$.  Now this subformula will only become blatantly false if both variables are given the same value, corresponding to them being painted the same color in \ruleset{Proper 2-coloring}.
	
	The overall formula is the conjunction of these pieces: $\BigAnd_{(i,j) \in E} (x_i \wedge \overline{x_j}) \vee (\overline{x_i} \wedge x_j)$
	
	Since the two games are equivalent, \ruleset{either-anywhere-same} is \cclass{PSPACE}-complete.
  \end{proof}

\section{Either-Anywhere-Different}
\label{section:either-anywhere-different}

The \ruleset{either-anywhere-different} ruleset consists of the games between the two players who can play either value on any unassigned variable each turn.  The players have separate goals: the even player is attempting to set the entire formula to true, the odd player is shooting for false.

\subsection{Sample Game}

As an example, consider the initial position with the same formula given in section \ref{section:sampleQBF}: $(\overline{x_0} \vee x_3 \vee \overline{x_1}) \wedge (x_2 \vee x_1 \vee \overline{x_6}) \wedge (x_4 \vee \overline{x_6} \vee x_0) \wedge (\overline{x_2} \vee \overline{x_4} \vee x_3)$.  One legal sequence of plays from the initial position could be:

\begin{itemize}
  \item Even/True assigns \textbf{T} to $x_3$:\\
    \begin{tabular}{l}
      $(\overline{x_0} \vee \textbf{T} \vee \overline{x_1}) \wedge (x_2 \vee x_1 \vee \overline{x_6}) \wedge (x_4 \vee \overline{x_6} \vee x_0) \wedge (\overline{x_2} \vee \overline{x_4} \vee \textbf{T})$ \\
      $ = (x_2 \vee x_1 \vee \overline{x_6}) \wedge (x_4 \vee \overline{x_6} \vee x_0)$
    \end{tabular}
  \item Odd/False assigns \textbf{T} to $x_6$:\\
    \begin{tabular}{l}
      $(x_2 \vee x_1 \vee \textbf{F}) \wedge (x_4 \vee \textbf{F} \vee x_0)$ \\
      $ = (x_2 \vee x_1) \wedge (x_4 \vee x_0)$ 
    \end{tabular}
  \item Even/True confidently assigns \textbf{T} to $x_4$:\\
	\begin{tabular}{l}
	  $(x_2 \vee x_1) \wedge (\textbf{T} \vee x_0)$ \\
	  $ = x_2 \vee x_1$ 
	\end{tabular}
  \item Odd/False assigns \textbf{F} to $x_1$:\\
	\begin{tabular}{l}
	  $x_2 \vee \textbf{F}$ \\
	  $ = x_2$
	\end{tabular}
  \item Even/True assigns \textbf{T} to $x_2$ and wins.
\end{itemize}

\subsection{\cclass{PSPACE}-completeness}

Although the reduction for this completness proof is quite simple, we first introduce an intermediate ruleset for clarity:

\begin{definition}[\ruleset{Toy Positive CNF}]
  \emph{\ruleset{Toy Positive CNF}} is exactly the same as \ruleset{Positive CNF}, except that each player may choose to assign either True or False to the variable on their turn.
\end{definition}

\begin{lemma}[\ruleset{Toy Positive CNF}]
  \ruleset{Toy Positive CNF} is \cclass{PSPACE}-complete.  
\end{lemma}

  To prove this, we will reduce from \ruleset{Positive CNF} (see Section \ref{section:positive-cnf}).

\begin{proof}
  To begin this proof, we notice that it never improves a player's strategy to choose to assign a variable with the value opposite their identity.  Since no negations exist in the formula, $f$, the True player never benefits by assigning False and the False player never benefits by choosing True.
  
  Thus, any winning strategy cooresponds directly to a winning strategy in \ruleset{Positive CNF}.  Our reduction is trivial; no transformation is needed to reduce from \ruleset{Positive CNF} to \ruleset{Toy Positive CNF} to prove completeness.
\end{proof}

\begin{theorem}[\ruleset{either-anywhere-different} is \cclass{PSPACE}-complete.]
  \ruleset{either-anywhere-different} is also \cclass{PSPACE}-complete.  
\end{theorem}

  To prove this, we will reduce from \ruleset{Toy Positive CNF}.
  
  \begin{proof}
    Note that the set of \ruleset{Toy Positive CNF} positions is exactly a subset of the set of \ruleset{either-anywhere-different} positions.  Thus, we can use the trivial (identity) reduction to show that \ruleset{either-anywhere-different} is also \cclass{PSPACE}-complete.
  \end{proof}

\section{Conclusions}

  This work defines seven new combinatorial game rulesets based on satisfying Boolean formulas.  Two of these are trivial in that the players do not have any choices to make.  The other five are all computationally difficult (\cclass{PSPACE}-complete) to determine which player has a winning strategy in the worst case.  These five are proved by simple reductions; two require no transformations on the positions.
  
  These five hard games differ in the style of allowed move options.  Each offers a different source problem for future \cclass{PSPACE}-hardness reductions.  By choosing the best of these rulesets that matches the target game, reductions for new rulesets could be easier to find.  In particular, two of the resulting \cclass{PSPACE}-hard games, \ruleset{Either-Local-Same} and \ruleset{Either-Anywhere-Same}, are also impartial and are good candidates to use to prove the hardness of new impartial games.

\section{Future Work}

  There is an obvious extension by adding more toggles to create more rulesets.  Another interesting open question is whether the hardness persists for certain types of formulas.
  
\subsection{Additional Toggles}

  This work can be expanded on by introducing more toggle properties for satisfiability games.  Schaefer defines many related games, including games with partitions on the variables between the two players \cite{DBLP:journals/jcss/Schaefer78}.  Certainly all relevant properties of rulesets are not covered here; expanding on the space of toggle properties will be very helpful for future \cclass{PSPACE}-hardness reductions for games.  Unfortunately, fully incorporating any new toggle requires doubling the number of overall rulesets.
  
\subsection{Open: Hardness in 3CNF}

  Unfortunately, three of the reductions used in this don't translate directly into conjunctive normal form (CNF) formulas with up to three literals (3CNF).  (Specifically, the reductions for \ruleset{By-Player-Anywhere-Same}, \ruleset{Either-local-same}, and \ruleset{Either-Anywhere-Same}.)  Although transformations to 3CNF exist, they may not preserve winnability in each of the rulesets.  Showing hardness in any of these would be a stronger result.

\bibliographystyle{plain}
\bibliography{/home/paithan/Dropbox/paithan}

\end{document}